\documentclass[journal]{IEEEtran}
\usepackage{microtype}
\usepackage{lipsum}
\usepackage{amsmath,amsfonts,amssymb,graphicx}
\usepackage{hyperref}
\usepackage{color}
\usepackage{xcolor}
\usepackage{mathrsfs}
\usepackage[utf8]{inputenc}
\usepackage[english]{babel} 
\usepackage{amssymb,amsmath,amsthm,latexsym,amscd,amsfonts,algorithm,algorithmicx,algpseudocode}
\usepackage{graphicx} 
\usepackage{tikz}
\RequirePackage{amsopn}
\RequirePackage{amsfonts}
\RequirePackage{amsthm}
\ifCLASSINFOpdf
\else
\fi

\newtheorem{lem}{Lemma}

\newtheorem{thm}{Theorem}
\pdfoutput=1
\title{Iterative Null-space Projection Method of Adaptive Thresholding in Sparse Signal Recovery and Matrix Completion}

\author{Ashkan~Esmaeili, Ehsan~Asadi, and~Farokh~Marvasti\\Advanced Communications Research Institute (ACRI), and\\Electrical Engineering Department, Sharif University of Technology, Tehran, Iran\\~aesmaili@stanford.edu, asadikangarshahi\_ehsan@ee.sharif.edu,~marvasti@sharif.edu
	\thanks{Ehsan Asadi, and Farokh Marvasti are with the Advanced Communications Research Institute (ACRI), and Department of Electrical Engineering, Sharif University of Technology. Ashkan Esmaeili was with Stanford Electrical Engineering Department, and is Stanford MS alumni. He is now on Ph.D program with Electrical Engineering Department at Sharif University of technology and ACRI.   }
	}

\begin{document}

\maketitle
\begin{abstract}
Adaptive thresholding methods have proved to yield high Signal-to-Noise Ratio (SNR) and fast convergence for sparse signal recovery. Recently, it was observed that the robustness of  a class of iterative sparse recovery algorithms such as Iterative Method with Adaptive Thresholding (IMAT) outperforms the well-known LASSO algorithm in terms of reconstruction quality, convergence speed, and the sensitivity to the noise. In this paper, we introduce a new method towards sparse signal recovery from random samples (RS) of the sensing matrix or its generalized version Compressed Sensing (CS) problem. The logic of this method is based on iterative projections of the thresholded signal onto the null-space of the sensing matrix.
The simulations results reveal that the proposed method has the capability of yielding noticeable output SNR values when the number of samples approaches twice the sparsity number, while other methods fail to recover the signals when approaching this number.  We have also extended our algorithm to Matrix Completion (MC) scenarios and compared its efficiency to other well-known approaches for MC in the literature.
\end{abstract}


\section{Introduction}
Compressed Sensing (CS) \cite{candes2008introduction}
\cite{candes2006compressive} \cite{foucart2013mathematical} has grown to become one of the most widespread problems of interest due to its various applications. Greedy algorithms such as Orthogonal Matching Pursuit (OMP) \cite{tropp2007signal} are among the very first algorithms used for sparse recovery. These algorithms have been dominated by robustness of algorithms introduced later on in the literature including LASSO \cite{tibshirani1996regression} and, IMAT \cite{Marvasti2012}, and RLS method\cite{holland1977robust} in terms of the reconstructed signal Signal-to-Noise Ratio (SNR). 
The general CS problem formulation is as follows:
\begin{align}
 \min & \quad  ||x||_0\\\nonumber
 \rm{s.t.} & \quad y=\Phi x,
\end{align}
where $\Phi$ is the sensing (measurement) matrix, $y$ is the measurement vector, and $x$ is the main signal assumed to be sparse. A special case of interest is RS where the measurement matrix has random zeros and ones in its diagonal \cite{marvasti2012nonuniform}. This formulation is for the scenario when the noise is absent. In the presence of the noise, the formulation changes as follows:
\begin{align}\label{eq:cs}
\min & \quad||x||_0\\\nonumber
 \rm{s.t.} & \quad||y-\Phi x||_2^2 \leq \epsilon.
\end{align}
The solution to under-determined linear system is the sparsest solution if $||x||_0\leq \frac{1}{1+\mu(\Phi)}$, where $\mu(\Phi)$ is the coherence of the matrix $\Phi$ \cite{azghani2014sparse}. Knowing this, we have developed a new method which converges to the sparsest solution knowing the fact that if the solution is unique, it will be the sparsest one. The intuition behind our method will be elaborated in section II. Now, we briefly go over the most well-known methods in sparse recovery in signal processing.
\subsection{Brief Review over Efficient Sparse Recovery Methods in the Literature}

We briefly review the methods we wish to consider in our comparisons. We briefly mention how they work and additionally their features and characteristics. Later on, in the simulations section we will also investigate the types of data we use and the accuracy, sensitivity to noise, and the required number of samples in order to recover the main signal. We will compare diverse aspects of the introduced methods to our proposed method called Iterative Null-space Projection Method with Adaptive Thresholding (INPMAT). In general, sparse recovery methods could be classified into three main groups: 1- $l_1$- minimization approaches, 2- Greedy methods, and 3- $l_0$ approximation methods.\\

\subsubsection{$l_1$-minimization approach}
\subsection*{Lasso}
We approximate $l_0$ - norm with its closest convex surrogate which is $l_1$- norm. Therefore, the resulting problem formulation is as follows:
\begin{align}\label{eq:csl1}
\min & \quad||x||_1\\\nonumber
\rm{s.t.} & \quad||y-\Phi x||_2^2 \leq \epsilon.
\end{align}
The dual of the problem (\ref{eq:csl1}) is the following, which is known as LASSO minimization problem:
\begin{align}\label{eq:lasso}
 \min & \quad||x||_1 + \lambda ||\Phi x-y||_2^2\\\nonumber
 \rm{s.t.} & \quad \lambda \geq 0
\end{align}
The LASSO is known to yield sparse solutions for $x$ in addition to resistance to noise.
In this section we review Lasso method. Lasso is the following minimization problem:
\begin{equation}
\min_x  ||y-\Phi x ||_2^2 +\lambda ||x||_1 
\end{equation}
 .The optimal $\lambda$ is obtained by sweeping the grid of values and minimizing error on the test set.
The solution to Lasso can be derived by many different implementations. One approach in solving the $l_1$-norm minimization problem is the ADMM method \cite{Boyd:2011:DOS:2185815.2185816}.\\

\subsubsection{Greedy approach}
\subsubsection*{OMP and CoSaMP}
In \cite{tropp2007signal}, we can find how OMP works. CoSaMP method is introduced in \cite{needell2009cosamp} as a greedy method in recovering sparse signal. These two are somehow similar to each other in terms of the structure; however, the CoSaMP is more complicated and could be considered as a generalized OMP. The general behavior of this method is that in high SNR input scenarios it can recover the signal provided that it has enough samples or measurements in the under-determined setting. These two methods require to know the sparsity level which is considered as a drawback in comparison to the INPMAT which is independent of knowing the sparsity level. In Analysis and simulations section, we will see that the INPMAT requires far fewer samples in order to recover the main signal than OMP and COSaMP.\\

\subsubsection{$l_0$-approximation}
\subsubsection*{IHT}
One can find how Iterative Hard Thresholding (IHT) works in \cite{blumensath2009iterative}. This method applies one update followed by one thresholding step. The thresholding is based on selecting the $s$ largest components; where $s$ is the sparsity number, and $s$ is known to IHT, or using a hard threshold in each iteration.

\subsubsection*{IMAT}
The Iterative Method with Adaptive Thresholding is introduced in \cite{Marvasti2012} by Marvasti et al. This method and its modified version IMATCS work as stated in \cite{azghani2014sparse}. One update step is followed by thresholding with some value which is shrinking during iterations exponentially. One can find two implementation versions of IMAT in \cite{imatsite}. In \cite{esmaeili2016comparison}, and \cite{esmaeili2016fast}, we have seen that IMATCS can outperform Lasso in some scenarios like dealing with missing data. IMAT has been shown to be profitable as in microwave imaging \cite{azghani2014microwave}. In general, ierative methods are a large class of CS methods \cite{chartrand2008iteratively}, from which we are focusing on IMAT.

\subsubsection*{SL0}
The SL0 method is first introduced in \cite{4663911}. It is based on approximating the $L_0$-norm with a smooth function. The implementations are also provided in the simulations section.

\subsection{Review over Matrix Completion (MC)}
We will also provide the extended version of our proposed algorithm to MC problems. MC has become popular due to the Netflix recommendation contest. MC for low rank matrices have been scrutinized recently by many authors, and various methods have been developed towards dealing with this problem such as \cite{keshavan2010matrix}, \cite{cai2010singular}, and \cite{mazumder2010spectral}. In general, MC is an extended version of sparse recovery in matrix domains. Considering certain constraints for sparse signals such as the number of required samples, or the rank of the sensing matrices lead to unique recovery of the sparse signal. The concept of low-rank matrices is similar to sparse matrices in some sense. In fact, low rank induces sparsity in matrix domain. In order to deal with matrix completion, the convex surrogate of the rank function is considered, and the problem to be solved is as follows:

\begin{equation}
X^*=\mathrm{argmin}_X ||P_E(X-A)||_2^2+\lambda||X||_*,
\end{equation} 
where $P_E(.)$ is the projection onto the observation index set, and $||X||_*$ is the trace norm of $X$. $X^*$ is supposed to be the reconstructed matrix. We use the concept of INPMAT in order to extend the problem to the matrix completion scenario. We mention the methodology briefly in section IV and provide related simulations in section VI.
The rest of the paper is organized as follows: In section II, we introduce INPMAT. In section III, we provide noise analysis. In section IV, we provide an extension of the method to matrix completion. In section V, we take a detour to provide a modified version of our algorithm which is derived by looking into our formulation from a different perspective and provide the heuristics. Finally, we conclude the paper in the last section.

\section{INPMAT Procedure}
 Now, we introduce our method. Algorithm \ref{tab:L2SAT} provides the procedure of INPMAT. Let $x^k$ denote the signal recovered after $k$-th step. Let $T^k$ denote the diagonal matrix which indicates the support of the recovered signal on its diagonal in the $k$-th iteration. $1$'s on the diagonal show that the element is non-zero and $0$'s show the opposite. Let $\mathbf{I}(.)$ denote the indicator function of the expression inside the parentheses, i.e, it returns $1$ if the expression inside the parentheses is true and returns $0$ otherwise.
First, we project the signal onto the subspace created by the eigenvectors of the diagonal matrix $T$ in each step. This is in fact equivalent to confining the signal index to the support determined by the nonzero elements of the diagonal of the matrix $T$. The projected vector is heuristically the closest solution with sparsity number equal to the dimension of the subspace formed by $T$ to the set $S_\Phi=\{x: y=\Phi x\}$. The reason is that knowing the support we only need to apply a pseudo-inverse to find the sparsest solution which has minimum local distance to $S_{\Phi}$. Indeed, we wish the components outside the support to shrink as much as possible. Next, we aim to project the resultant image onto $S_\Phi$, 
i.e. we wish to find the solution to be in the set of points holding in the constraint $y=\Phi x$ .
 That is why we project the thresholded signal onto the $S_{\Phi}$. We iteratively continue these steps until the stopping criterion holds. The stopping criterion is as follows:
\begin{equation}
||y-\Phi x^k||_2^2 \geq \epsilon
\end{equation}

Now, we proceed to mention the concept of adaptive thresholding that we use in this paper. The constraint $y=\Phi x$ forces the solutions of $l_0$-norm minimization to fall in the translated null-space of $\Phi$, i.e, $S_\Phi$. Simultaneously, we are looking forward to finding the vector in this set which is the sparsest solution. 
\subsection{Geometric Interpretation}
We consider all $s$-dimensional subspaces, and we show that projecting $x$ onto each of these subspaces followed by projecting back onto the $S_{\Phi}$ reaches a new solution of $y=\Phi x$ whose components outside the support shrink, i.e. $||(I-T)x^{k+1}||_2^2 \leq ||(I-T)x^{k}||_2^2$. In addition, we also desire to make the solution sparser. Thus, we regularize the residual defined above with $Tr(T)$ to make the solution sparser. As a result, the objective function we minimize is as follows:
\begin{equation}\label{obj}
\begin{split}
\min_{\buildrel {T,} \over {x \in \Phi ^\dagger y+\rm{N(\Phi)}}}  & \quad  ~||(I-T)x||_2^2+\lambda Tr(T)\\
\end{split}
\end{equation}
We denote the objective in \ref{obj} with $f(x,t).$ 
\begin{thm}\label{thm1}
Two-step projections onto the $k$-dimensional support subspace, and $S_\Phi$ leads to decrease in the contraction operator $||(I-T)x^{k}||_2^2$, i.e, $||(I-T)x^{k+1}||_2^2 \leq ||(I-T)x^{k}||_2^2.$
\end{thm}
\begin{proof}
By Pythagorean theorem we have, 
\begin{equation}
||(I-T)x^{k}||_2^2=||x^k-x^{k+1}||_2^2+||x^{k+1}-Tx^{k}||_2^2
\end{equation}
Similarly, 
\begin{equation}
||x^{k+1}-Tx^k||_2^2=||Tx^k-Tx^{k+1}||_2^2+||(I-T)x^{k+1}||_2^2
\end{equation}
Thus, 
\begin{align}
& ||(I-T)x^k||_2^2 = \\ \nonumber
& ||T(x^k-x^{k+1})||_2^2+||(I-T)x^{k+1}||_2^2+||x^k-x^{k+1}||_2^2
\end{align}

\begin{equation}
\Rightarrow
||(I-T)x^k||_2^2 \geq ||(I-T)x^{k+1}||_2^2
\end{equation}
\end{proof}
\rm{First, we argue that if instead of using pseudo-inverse, we fix $T$, and project $x$ onto $S_\Phi$ and $T$ iteratively, the result of iterations converges to the solution by the pseudo-inverse \cite{amini2008convergence}. Iterative methods do not have the complexity of computing pseudo-inverse.} Theorem \ref{thm1} proves how the iterative method which could be used instead of pseudo-inverse method converges to a locally optimal point.\\

\subsection{Algebraic Interpretation}

\begin{algorithm} 
		\caption{INPMAT}\label{tab:L2SAT}
		\begin{algorithmic}
			\State \textbf{Input:}
			\State {A measurement matrix} ${\mathbf \Phi} \in \mathbb{R}^{m\times n}$
			\State {A measurement vector } ${\mathbf y}\in \mathbb{R}^{m}$
			\State \textbf{Output:}
			\State {A recovered estimate }$\widehat{{\mathbf x}}\in \mathbb{R}^{n}${ of the original signal.}
			
			\Procedure{INPMAT}{${\mathbf y},{\mathbf \Phi},{\mathbf x}$}
			\State ${\mathbf x}^{0}\gets \Phi^\dagger y$
					\State $k \leftarrow 0$
					\State $thr \leftarrow \max(|x^0|)$
			\While {$||y-\Phi x^k||_2^2 \geq \epsilon$}
						\State $T^{k}\leftarrow \mathbf{diag}(\mathbf{I}(|x^{k}|\geq thr)) $
			\State $s^{k}=(\Phi T^k)^\dagger y $\label{kkkk}		
			\State $x^{k+1}\leftarrow x^0+ (I-\Phi^\dagger\Phi)s^k$	\label{kkk}	

		    \State $thr \leftarrow \max(|(I-T^k)x^{k}|)$
			\State $k \leftarrow k+1$
			\EndWhile 
			\State \textbf{return} $\quad \widehat{{\mathbf x}}\gets {{\mathbf x}}^{k} $
			\EndProcedure
		\end{algorithmic}
\end{algorithm}

\begin{lem}
Let $\zeta$ denote the set of all diagonal matrices with diagonal entries in the interval $[0,1]$. The objective in \eqref{obj} is convex w.r.t $T$ over $\zeta$ and also w.r.t $x$. 
\end{lem}
\begin{proof}
\rm{Convexity w.r.t $x$ is obvious since the objective is quadratic in $x$ if we fix $T$. Now We assume that $x$ is fixed. The $Tr$ is a convex operator; therefore, it is enough to show that the first term is convex in $T$.
 Let $t=[t_1,t_2,....t_n]$ denote the diagonal of $T$. Now, if we define $h(t)=||(I-T)x||_2^2$ then the Hessian could be calculated as  $H_{ii}(h([t]))=x_i^2 , H_{ij}=0 ,~ \forall i \neq j.$ Thus the Hessian is positive-definite and the proof is complete.}

\end{proof}
It is obvious that setting $\lambda$ to $\infty$ leads to minimizing the $l_2$-norm of the signal $x$. Thus, we start with the best $l_2$-norm solution for $x$, i.e. setting $x=\Phi ^\dagger y$. Then, we shrink $\lambda$ and find the solution for each problem and then set the updated solution as the starting point for the next iteration which leads to sparsifying the sequence of solutions for each problem.
\begin{thm}\label{th2}
	Let the minimum distance between $S_\Phi$ and all the subspaces induced by the set of matrices in $\zeta$ which do not intersect $S_\Phi$ be denoted by $\epsilon$. Also let $k$ be the sparsity number of the sparsest solution to \ref{obj}, then the optimal solution of the problem \ref{obj} is the sparsest (best $l_0$- norm) solution if $ 0< \lambda < \frac{\epsilon^2}{k}$.
\end{thm}

\begin{proof}
	if $T \in \zeta$ intersects with $S_\Phi$, then there is a solution to the $y=\Phi x$ which is inside the subspace induced by $T$. We have assumed that the sparsity number for solution is $k$, and the solution is unique, therefore this intersection point has sparsity number larger than $k$. Thus,
$\lambda Tr(T) \geq \lambda k$ and as a result $f(x,T) \geq f(x^*,T^*)= \lambda k$, where $x^*$ is the sparsest solution and $T^*$ is the support of $x^*$.\\
If $T \in \zeta$ does not intersect $S_\Phi$, then $||(I-T)x||_2^2 \geq \epsilon^2$. We have assumed that $\lambda \leq \epsilon^2/k$. Thus, $f(x,T) > ||(I-T)x||_2^2 \geq \epsilon = \frac{\epsilon^2}{k}k > \lambda k = f(x^*,T^*).$
\end{proof}
\subsection{INPMAT Implementation Knowing Sparsity Number}
We can easily adapt the INPMAT to the case when we know sparsity. We do this by changing the stopping criterion in Algorithm \ref{tab:L2SAT} to a new criterion which stops if the number of iterations exceeds the sparsity number. Finally, knowing sparsity number, we can confine the output signal to the support formed by the $s$ largest components. We compared the performance of INPMAT in this case with the previous methods when they also have the knowledge of the sparsity number.

\section{Noise Analysis}
In this section, we analyze how the reconstruction quality varies w.r.t noise power. 
In the general noisy model, we can assume the CS model is modified as follows: 
\begin{equation}
y=\Phi x + \epsilon
\end{equation}
Let $\epsilon$ denote the i.i.d Gaussian noise vector which is added to the observation vector. We assume $\epsilon \sim N(0,\sigma^2 I)$.
\begin{thm}
We assume $\Phi$ is RIP for $k$-sparse signals with constant $\delta_k$, i.e. $\forall x \in R^n: (1-\delta_k)||x||_2 \leq ||\Phi x||_2 \leq (1+\delta_k)||x||_2$, then the variance of the output noise could be bounded as:
\begin{equation}
\frac{1}{1+\delta_k}\sigma^2 \leq Var(\epsilon_{out}) \leq \frac{1}{1-\delta_k}\sigma^2
\end{equation}
\end{thm}
\begin{proof}
Since $\Phi$ is R.I.P for $k$-sparse signals with constant $\delta_k$, we can conclude that 
\begin{equation}\label{gf}
\frac{1}{1+\delta_k}||x||_2 \leq ||\Phi ^\dagger x||_2 \leq \frac{1}{1-\delta_k}||x||_2
\end{equation}
It is worth noting that the property in \ref{gf} is held on each and every $k$ columns of $\Phi$ and $\Phi ^ \dagger$ since we are working with $k$- sparse signals. Therefore, we can claim that if the sensing matrix $\Phi$ has R.I.P condition, then the \ref{kkkk}-th step in Algorithm \ref{tab:L2SAT}, is acting on at most $k$-sparse signals, and as a result we can establish the following bounds for noise terms at the output.
 \begin{equation}\label{ex}
\frac{1}{1+\delta_k}||\epsilon||_2 \leq ||\Phi ^\dagger \epsilon||_2 \leq \frac{1}{1-\delta_k}||\epsilon||_2
\end{equation}
The eigenvalues of any projection matrix are either equal to $0$ or $1$.
Therefore, using the \ref{kkk}-th step of Algorithm \ref{tab:L2SAT} and applying triangle inequality, we have the following bound for the noise term in $x^k$:
\begin{equation}
0 \leq ||\epsilon_{out}||_2 \leq ||\epsilon||_2+\frac{1}{1-\delta_k}||\epsilon||_2=\frac{2-\delta_k}{1-\delta_k}||\epsilon||_2
\end{equation}
\end{proof}
We notice that the threshold is non-increasing and as a result, we always work with at most $k-$sparse signals. In other words, the recovered support size never exceeds the sparsity of the main signal.
Thus, we always can take advantage of the RIP property and therefore, we can establish the following bound:
\begin{equation}
Var(\epsilon_{out}) \leq \left(\frac{2-\delta_k}{1-\delta_k}\right)^2 \sigma^2 ,
\end{equation}
we have shown in theorem \ref{th2} that the recovered signal converges to the original signal. If we denote the initial SNR as $SNR_0$, then for the output SNR we have the following:
\begin{equation}
SNR_{out} \geq SNR_0 + 20log\left(\frac{1-\delta_k}{2-\delta_k}\right),
\end{equation}
which shows that for small values of $\delta_k$ the recovery SNR does not fall off $6$dB of the initial SNR; however, it could be very large depending on the sensing matrix and could even lead to large values of SNR in sparse recovery according to the reconstruction power of the algorithm. The achieved bound; however, is simply a lower bound. In simulations, we will see that the reconstructed SNRs are not necessarily sticking to this bound and are usually larger than this bound.    
\subsection{Noise alleviation using Tikhonov regularization}
When the matrix $\Phi$ is ill-posed, the noise term in the initiation point can become significant. The reason is that if we select the least square minimizer of $||y - \Phi x||_2^2$ as our starting point, i.e.:
\begin{equation}
x^0=\Phi^ \dagger (\Phi x + \epsilon).
\end{equation}
we observe that the term $\Phi ^ \dagger \epsilon$ (noisy component) would be noticeable if the condition number of sensing matrix is large. In order to relieve the effect of noise in the starting point, we use the Tikhonov regularization, i.e. we use the solution to the following general minimization problem as the starting point:
\begin{equation}\label{eq21}
x^0=\underset{x}{\mathrm{argmin}}~~||y-\Phi x||_2^2 + \mu ||\Gamma x||_2^2,
\end{equation}
where $\Gamma$ in \ref{eq21} is set to $I_{n\times n}$.
$\mu$ is the regularization parameter, the choice of which is explained in \cite{vogel1996non}, \cite{choi2007comparison}, \cite{hochstenbach2010iterative}.
\begin{figure}
	\centering
	\includegraphics[width=1\linewidth]{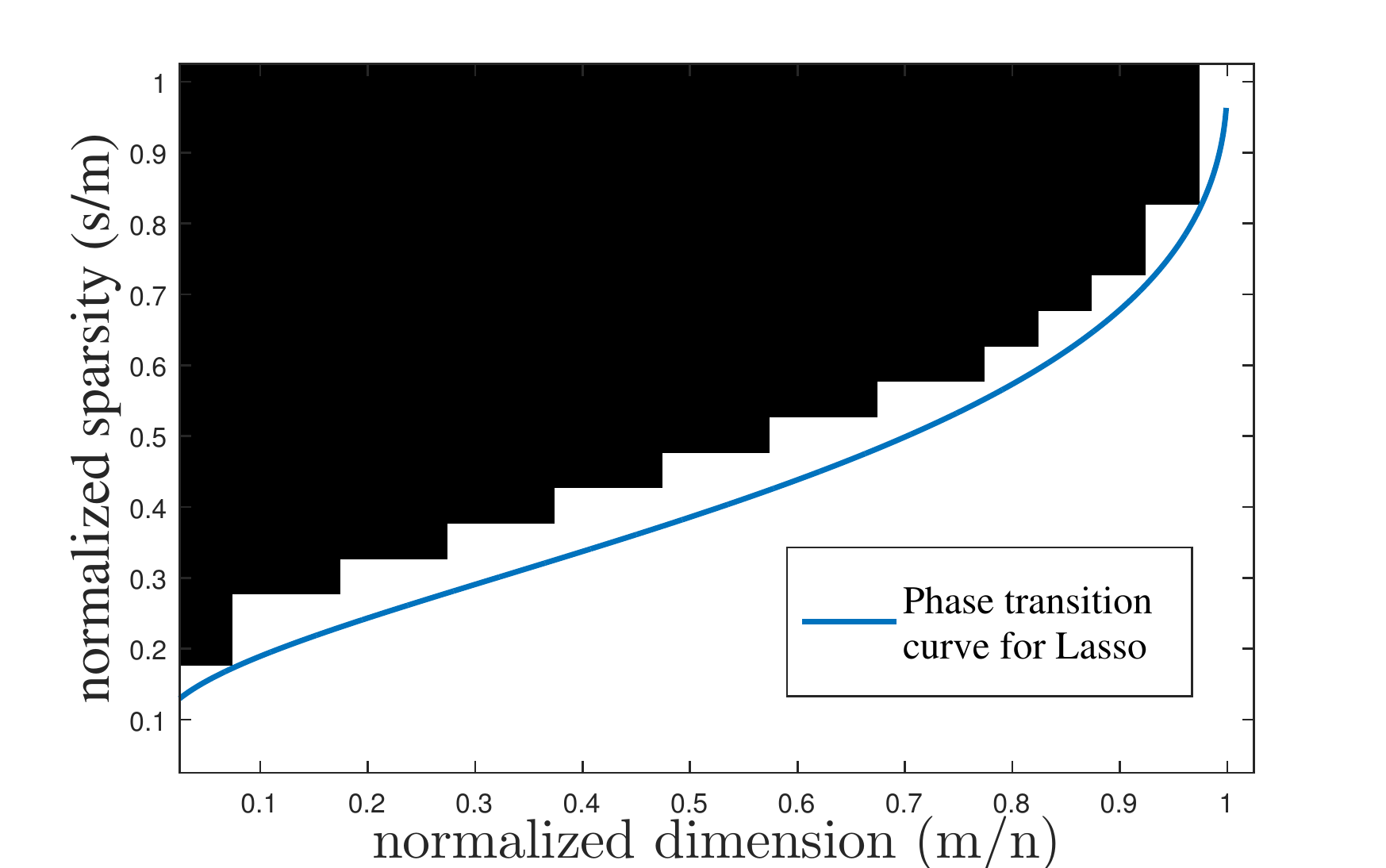}
	\caption{The Comparison of INPMAT phase transition curve with Lasso phase transition curve. The horizontal axis is normalized sparsity number. The vertical axis is normalized dimension.}\label{fig:ashkan1}
\end{figure}
\section{Extention to Matrix Completion}
In \cite{goldberg2010transduction}, Goldberg et al have improvised a method for matrix completion for low rank scenarios. The thresholding on singular values in their setting is carried out in an adaptive fashion. In summary, they have made an update at each iteration and thresholded the singular values in order to make the eigenvalues vector more sparse. Now, we provide intuition on how to extend INPMAT to matrix Completion problems. We call the extension of our algorithm to this scenario: Matrix Imputation with method of Adaptive Thresholding (MIMAT), where $\sigma_c(X)$ is the $c$-th singular value of matrix $X$, $P_\Omega$ is the projection operator onto the observed entries index set, and $P_{\Omega^\bot}$ is the projection onto the compliment of the set $\Omega.$
\begin{algorithm} 
	\caption{MIMAT}\label{MIMAT}
	\begin{algorithmic}[1]
		\State \textbf{Input:}
		\State {Initially observed matrix} ${\mathbf A}$
		\State {Stopping Criteria} ${\mathbf \epsilon_1, \epsilon_2}$
		\State \textbf{Output:}
		\State {A completed matrix}$\widehat{{\mathbf X}}${ of the original matrix}
		
		\Procedure{MIMAT}{${\mathbf A},{\mathbf \epsilon_1},{\mathbf \epsilon_2}$}
		\State $\Omega \leftarrow \{(i,j):\mathbf{A_{ij}}\neq 0 \}$
		\State $ c \leftarrow 1$
		\State $k \gets 0$
		\State $\mu_0 \gets ||A||_2$
		\State $Z_0 \gets \underline{0}_{m\times n}$
		\While {$||P_\Omega(\mathbf{A}-\mathbf{Z_k)}||_F^2 \geq \epsilon_1$}
		\While {$||\mathbf{X}_{k+1}-\mathbf{X}_{k}||_F^2 \geq \epsilon_2 $}
		\State $\mathbf{[U \Sigma V^T]} \leftarrow \rm{SVD} (\mathbf{X_k})$
		\State $\mathbf{X_{k+1}} \leftarrow \mathbf{U} {\max}(\mathbf{\Sigma}-\mu_k I,0) \mathbf{V^T}$
		\State $Z_{k+1} \gets X_{k+1}$
		\State $\mu_{k+1} \leftarrow \sigma_{c}(\mathbf{X_{k+1}}) $		
		
		\State $\mathbf X_{k+1} \leftarrow \mathbf P_\Omega(A) + \mathbf P_{\Omega^\bot} \mathbf (X_{k+1})$
		\State $k \leftarrow k+1$
		\EndWhile
		\State $ c \leftarrow c+1$
		\EndWhile
		\State \textbf{return} $\quad \widehat{{\mathbf X}}$
		\EndProcedure
	\end{algorithmic}
\end{algorithm}

%
%
\section{Simulation Results}
We implemented our simulations for two general cases. First, we assumed we know the sparsity number, and implemented all the methods including ours. In this case, we provided two types of figures. First, we fixed the input SNR, and varied the number of measurements, and plotted the output SNR versus the number of measurements. In the other type of figure, we plotted the output SNR versus input SNR while we have fixed the number of measurements. We also plotted all the figures mentioned above for the case when we do not have the knowledge of the sparsity number. 
We briefly summarize our observations as follows:
We observe that working with small number of measurements (approaching twice the sparsity number), our method outperforms other methods in reconstruction as shown in Figures \ref{fig:ashkan2}, \ref{fig:ashkan3}, \ref{fig:ashkan4}. When we increase the measurements number, all methods perform similarly and give high SNR values. This could be observed in Figures \ref{fig:ashkan6}, and \ref{fig:ashkan9}. In Figure \ref{fig:ashkan6}, the difference between the performance of different methods does not exceed 10dB. Therefore, the dominance of INPMAT shows itself for low sampling rates as in Figures \ref{fig:ashkan2}, \ref{fig:ashkan3}, \ref{fig:ashkan4}, and \ref{fig:ashkan7}. In Figures \ref{fig:ashkan2} and \ref{fig:ashkan3}, the input SNR is fixed to be 40dB. The output SNR is computed for all methods versus sampling rate. INPMAT reaches the highest SNR values. INPMAT reconstructs the signal with high SNR even for about 20\% sampling rate, while the performances of other methods fall for even larger sampling rates.\\
In Figure \ref{fig:ashkan1}, we notice that the phase transition curve for our method falls above the one for Lasso, i.e. the lasso phase transition curve falls below 50\% success rate phase transition curve of INPMAT \cite{tanner}, which shows the resistance of our method in recovery to noise level in the input.
In Figure \ref{fig:ashkan10} ,we compared MIMAT to Soft-Impute, SVT, and SL0. Our method outperforms SVT and Soft-Impute as well as SL0. The performances of MIMAT and SL0 are approximately similar for missing data percentage less than 60\%. However, MIMAT outperforms SL0 for missing data percentage greater than 80\% as provided in \ref{fig:ashkan10}. The Root Mean Squared Error (RMSE) for MIMAT is 0.2 less than RMSE of SL0. Soft-Impute never reaches RMSE less than 0.2. The SVT performance is also worse than MIMAT for the same missing percentages. For instance, for 65 \% missing percentage, the MIMAT reaches RMSE of 0, yet SVT has RMSE around 0.3.
 In Figures \ref{fig:ashkan11},~\ref{fig:ashkan12},~\ref{fig:ashkan13}, and~\ref{fig:ashkan14} we have compared the reconstruction power of the proposed method in various sampling rates for the two sensing methods of random sampling and the Gaussian measurenments. We can clearly see that when the available measurements are small, random sampling method has the capability of better reconstruction than Gaussian measurment. In Figure \ref{fig:ashkan11}, where the sampling rate is 25\%, we can see that the Gaussian measurment method of sensing fails to reconstruct the signal while RS reconstructs the signal and achieves 80 dB as the output SNR for the input SNR of 160 dB. In sampling rate of 37.5 \%, the Gaussian measurement method performs better but still cannot achieve the performqnce by RS. For instance, in Figure \ref{fig:ashkan12}, for input SNR of 150 dB, RS sensing method performs 70dB above Gaussian measurement method. For sampling rates above 50 \%, their performance tend to become similar and both achieve similar SNR values as in Figures \ref{fig:ashkan13}, and \ref{fig:ashkan14}.\\
The computational complexity for different algorithms were also measured in deifferent settings and the runtimes did not show a same pattern, and varied based on the used setting. Thus, we refrain from providing runtime table for a specific setting due to the variable behavior. In general, the runtime by INPMAT and MIMAT were comparable to other methods. The runtime achieved by INPMAT for example, never exceeded twice the minimum runtime achieved by other methods in msecs. 
\begin{figure}
	\centering
	\includegraphics[width=1\linewidth]{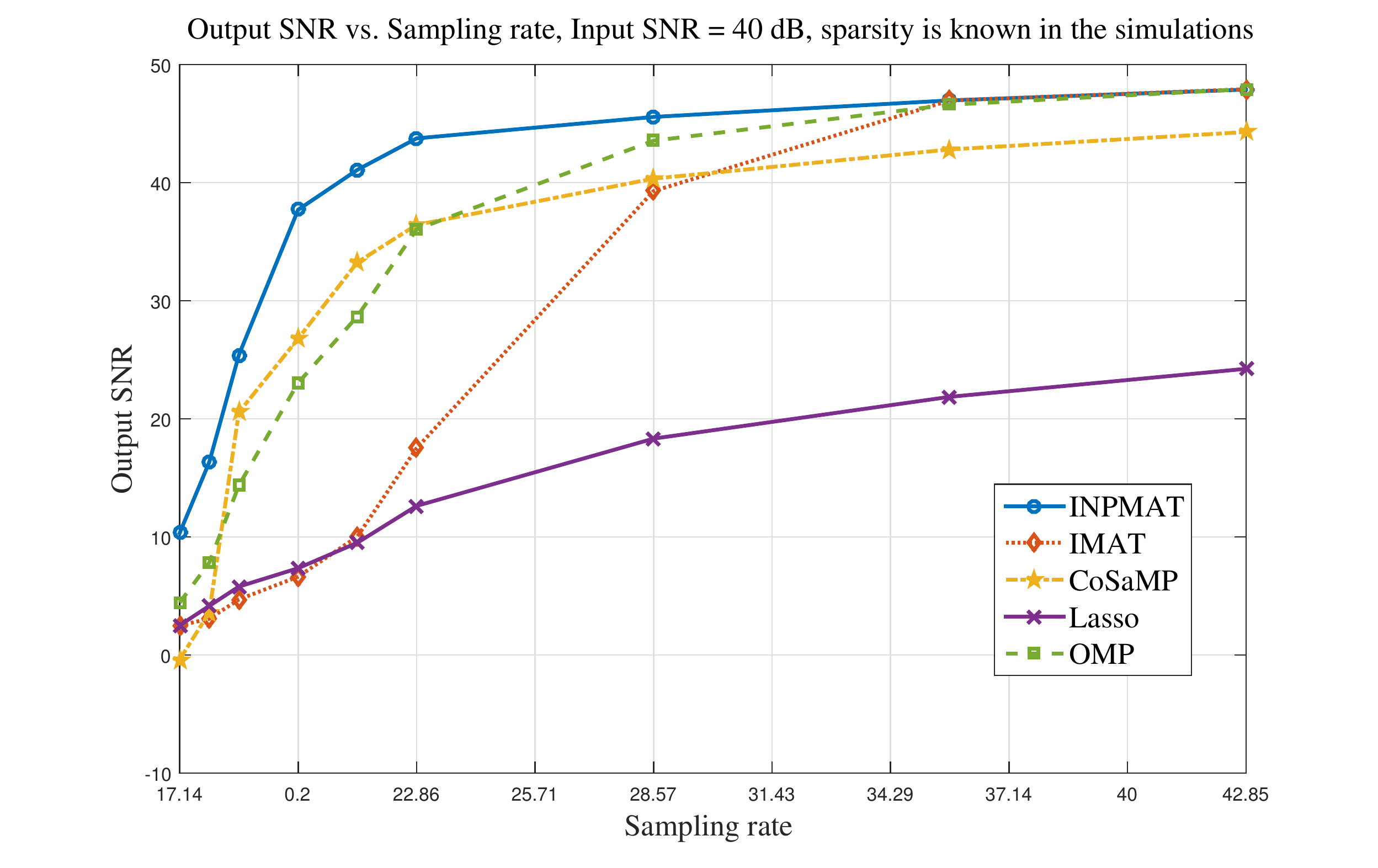}
	\caption{the signal size is $n=700$, sparsity equal to $40$.}\label{fig:ashkan2}
\end{figure}

\begin{figure}
	\centering
	\includegraphics[width=1\linewidth]{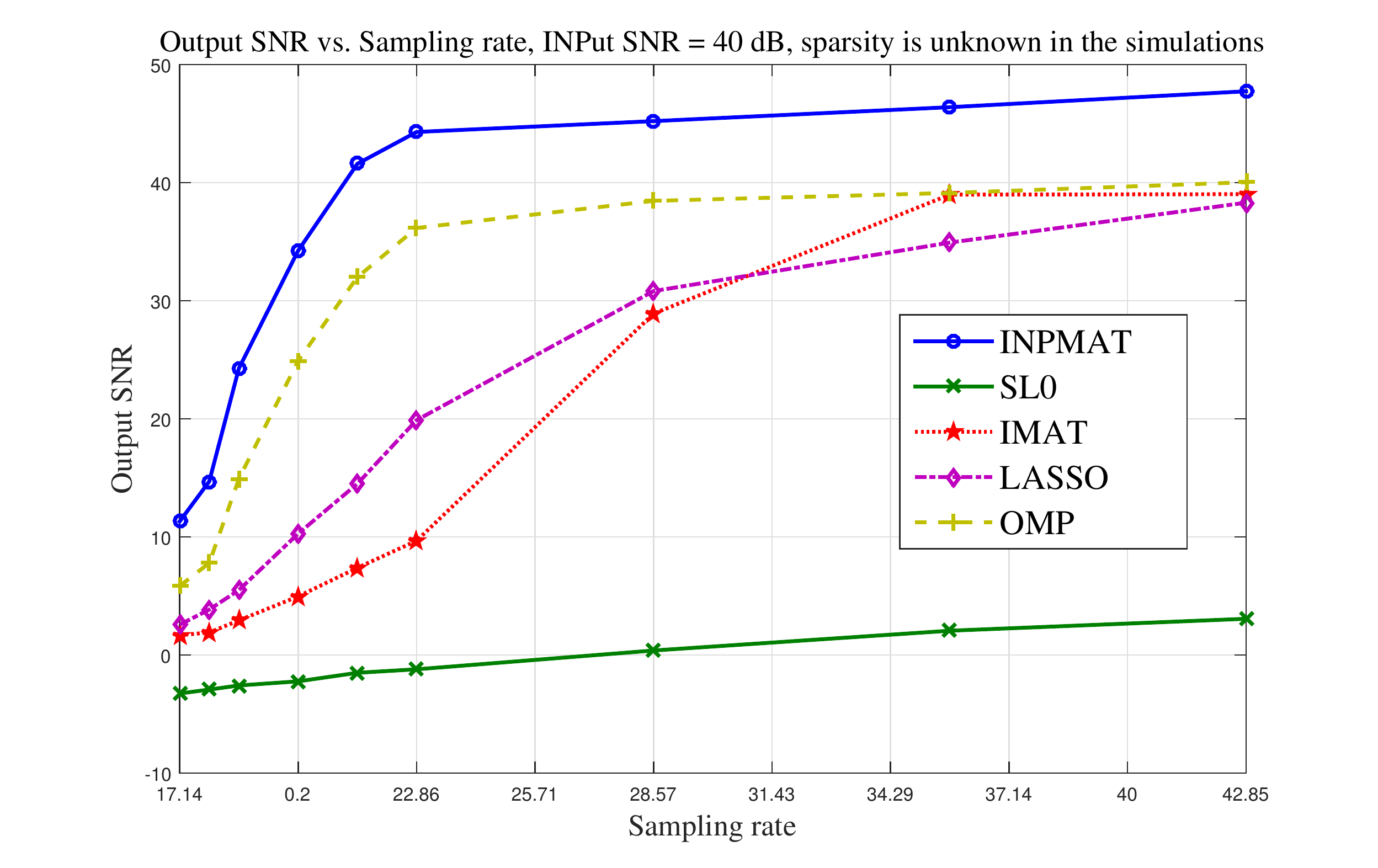}
	\caption{the signal size is $n=700$, sparsity equal to $40$.}\label{fig:ashkan3}
\end{figure}
\begin{figure}
	\centering
	\includegraphics[width=1\linewidth]{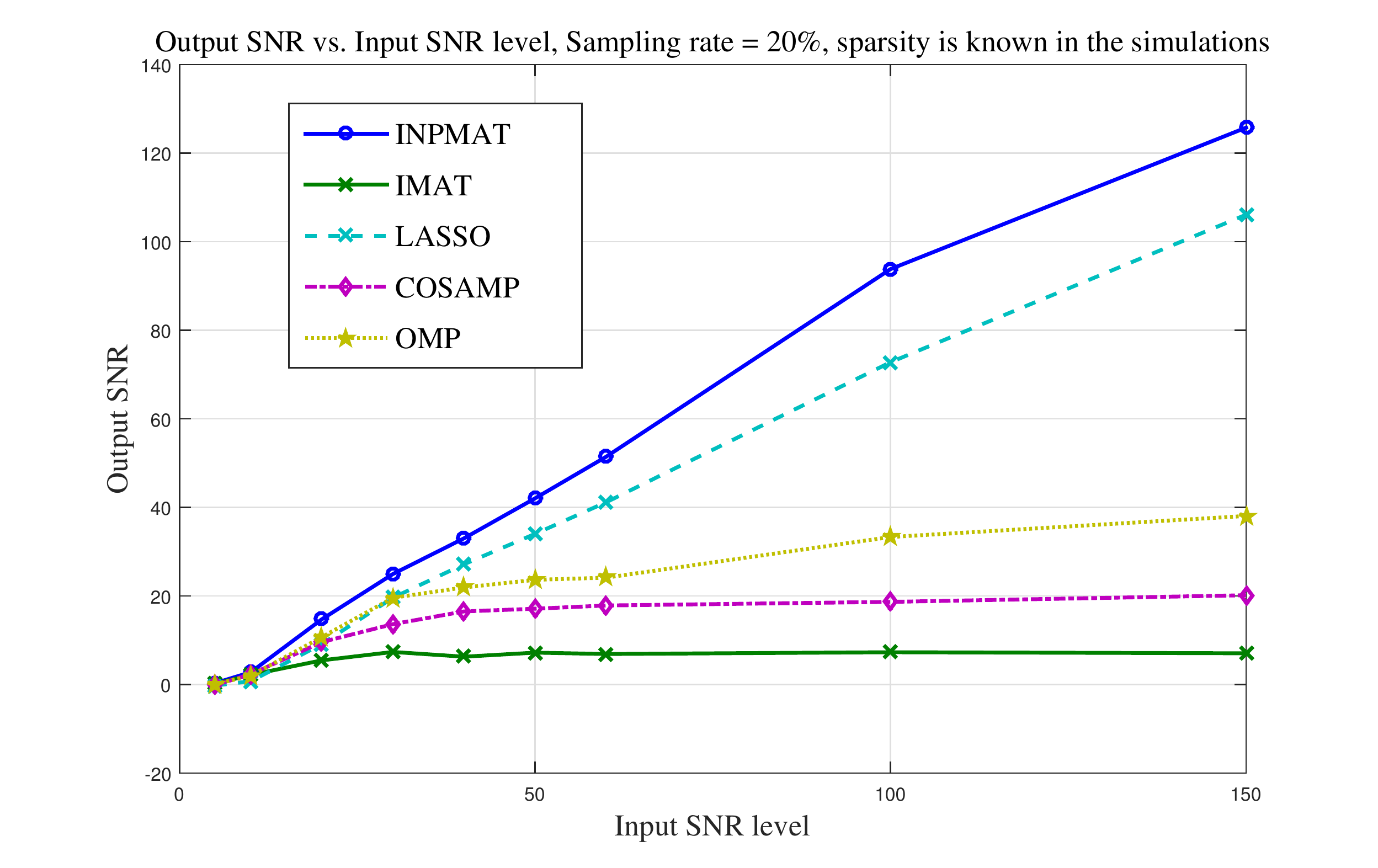}
	\caption{the signal size is $n=700$, sparsity equal to $40$.}\label{fig:ashkan4}
\end{figure}
\begin{figure}
	\centering
	\includegraphics[width=1\linewidth]{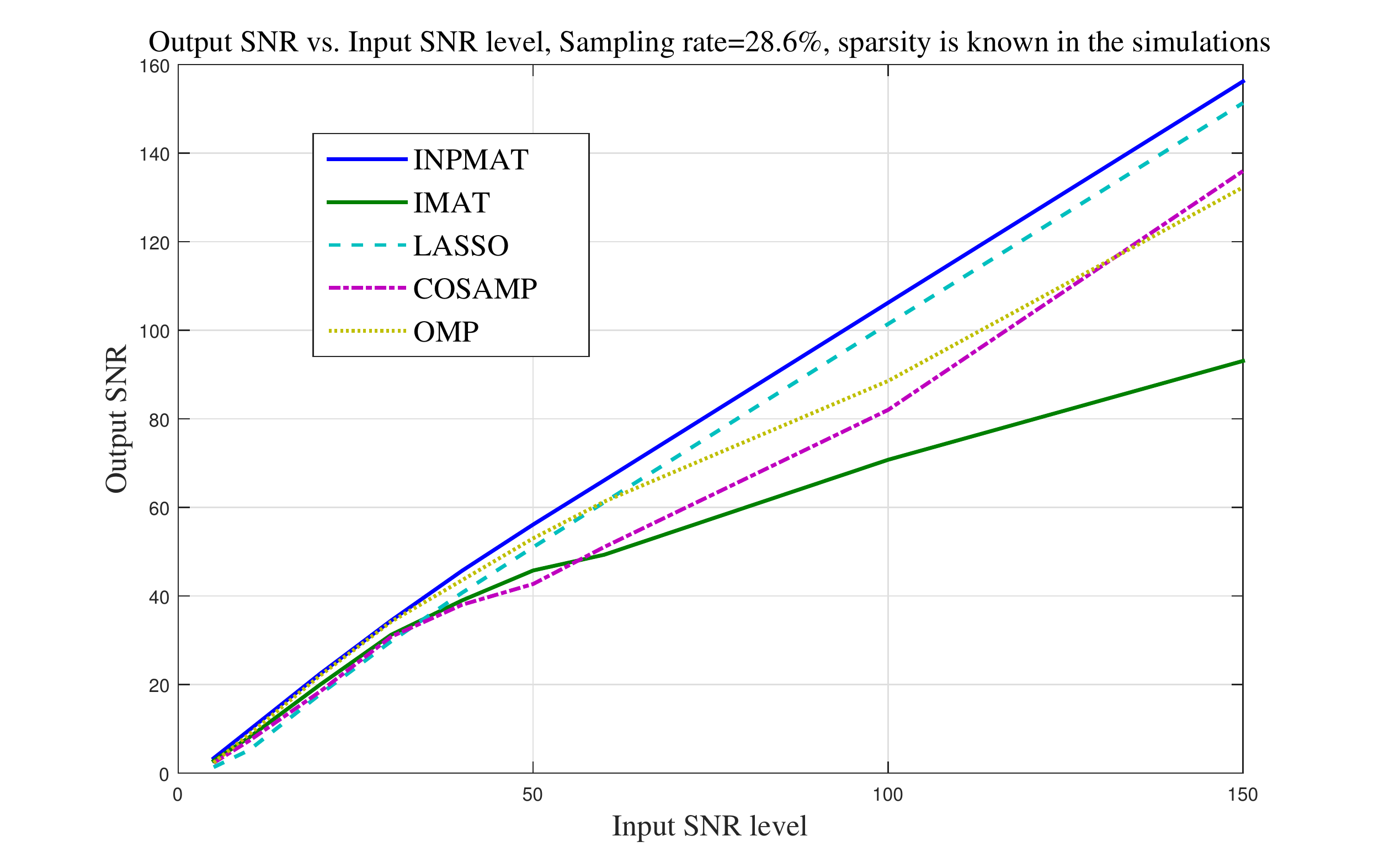}
	\caption{the signal size is $n=700$, sparsity equal to $40$.}\label{fig:ashkan5}
\end{figure}
\begin{figure}
	\centering
	\includegraphics[width=1\linewidth]{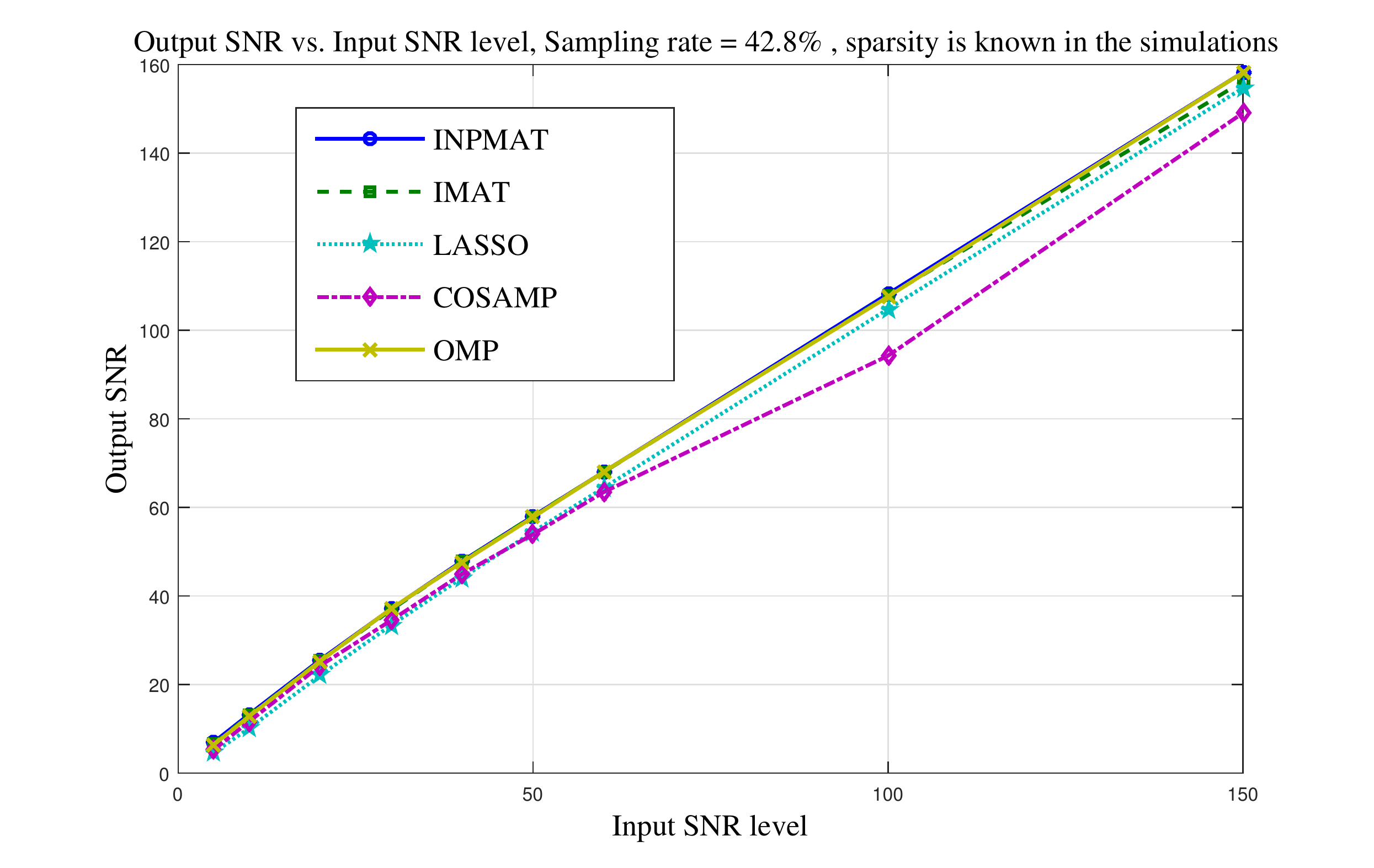}
	\caption{the signal size is $n=700$, sparsity equal to $40$.}\label{fig:ashkan6}
\end{figure}
\begin{figure}
	\centering
	\includegraphics[width=1\linewidth]{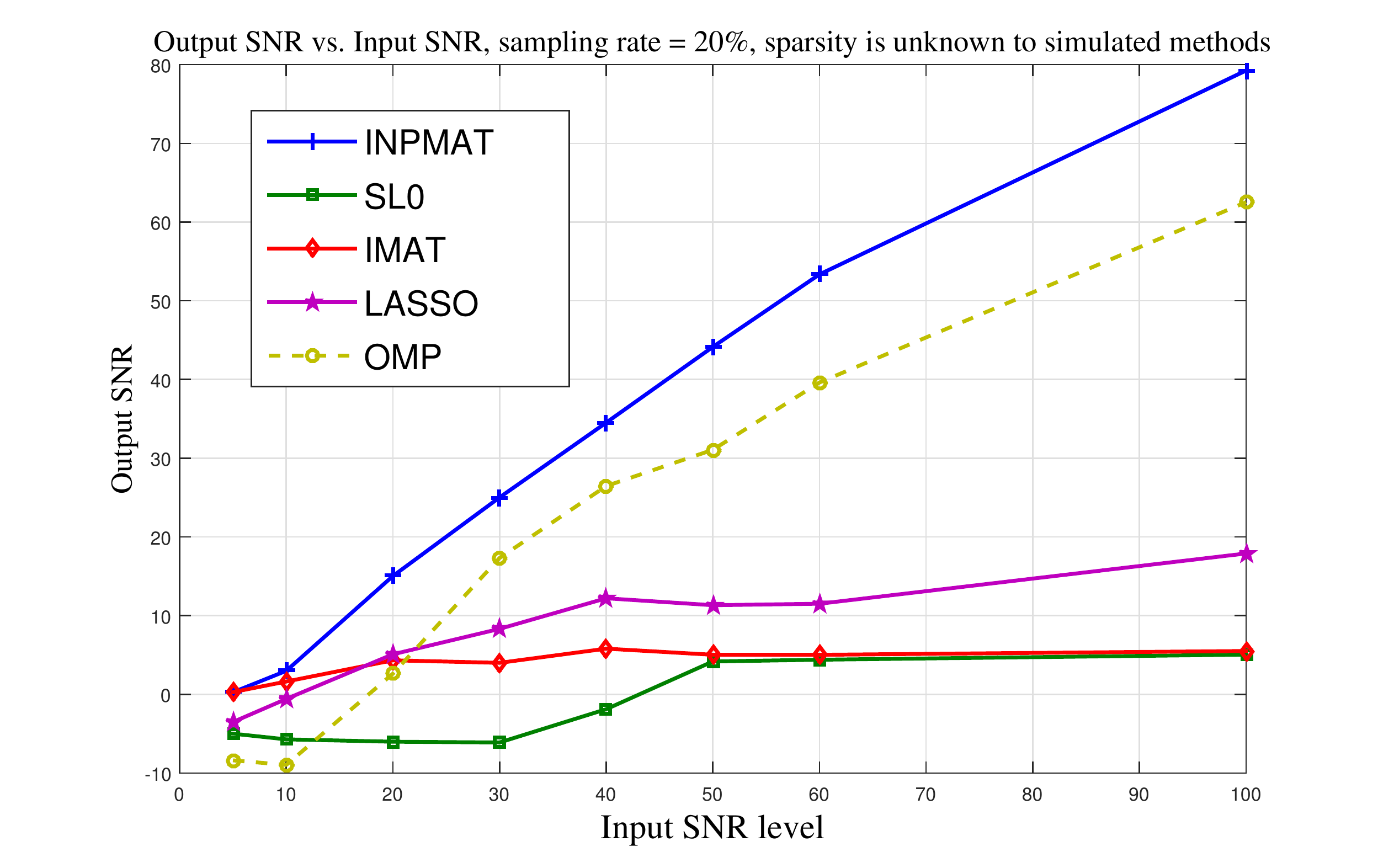}
	\caption{the signal size is $n=700$, sparsity equal to $40$.}\label{fig:ashkan7}
\end{figure}
\begin{figure}
	\centering
	\includegraphics[width=1\linewidth]{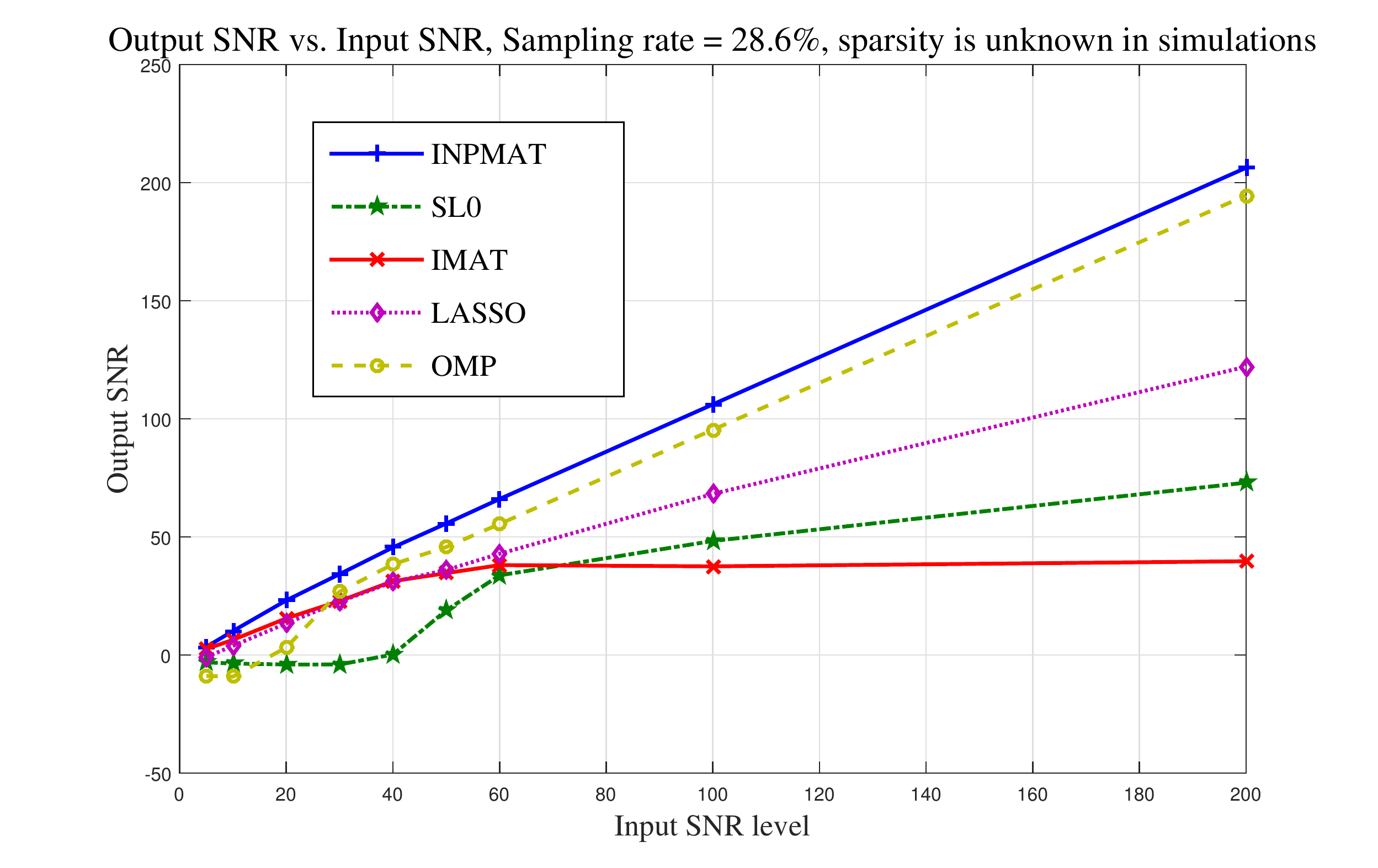}
	\caption{the signal size is $n=700$, sparsity equal to $40$.}\label{fig:ashkan8}
\end{figure}
\begin{figure}
	\centering
	\includegraphics[width=1\linewidth]{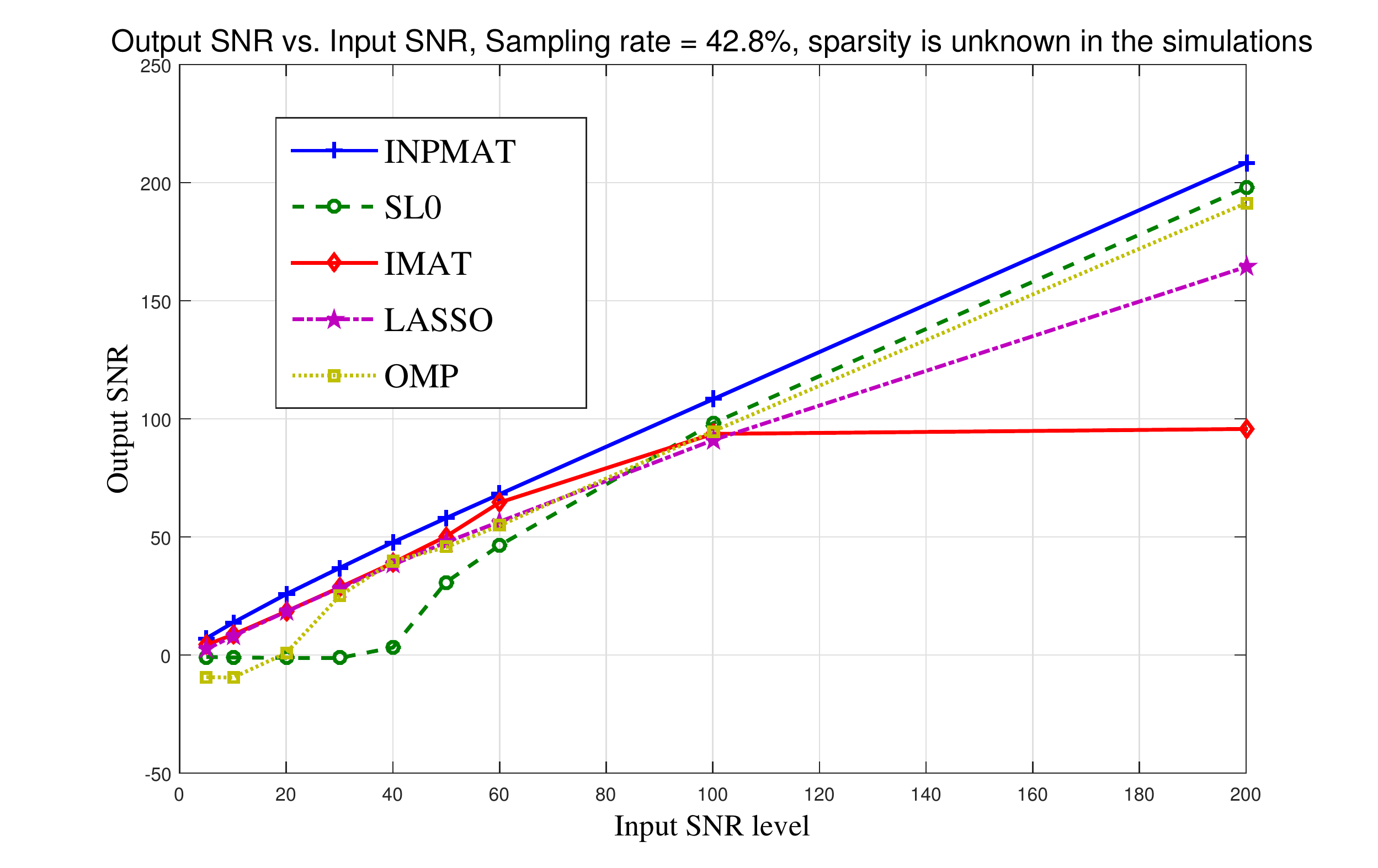}
	\caption{the signal size is $n=700$, sparsity equal to $40$.}\label{fig:ashkan9}
\end{figure}
\begin{figure}[h!]
	\centering
	\includegraphics[width=1\linewidth]{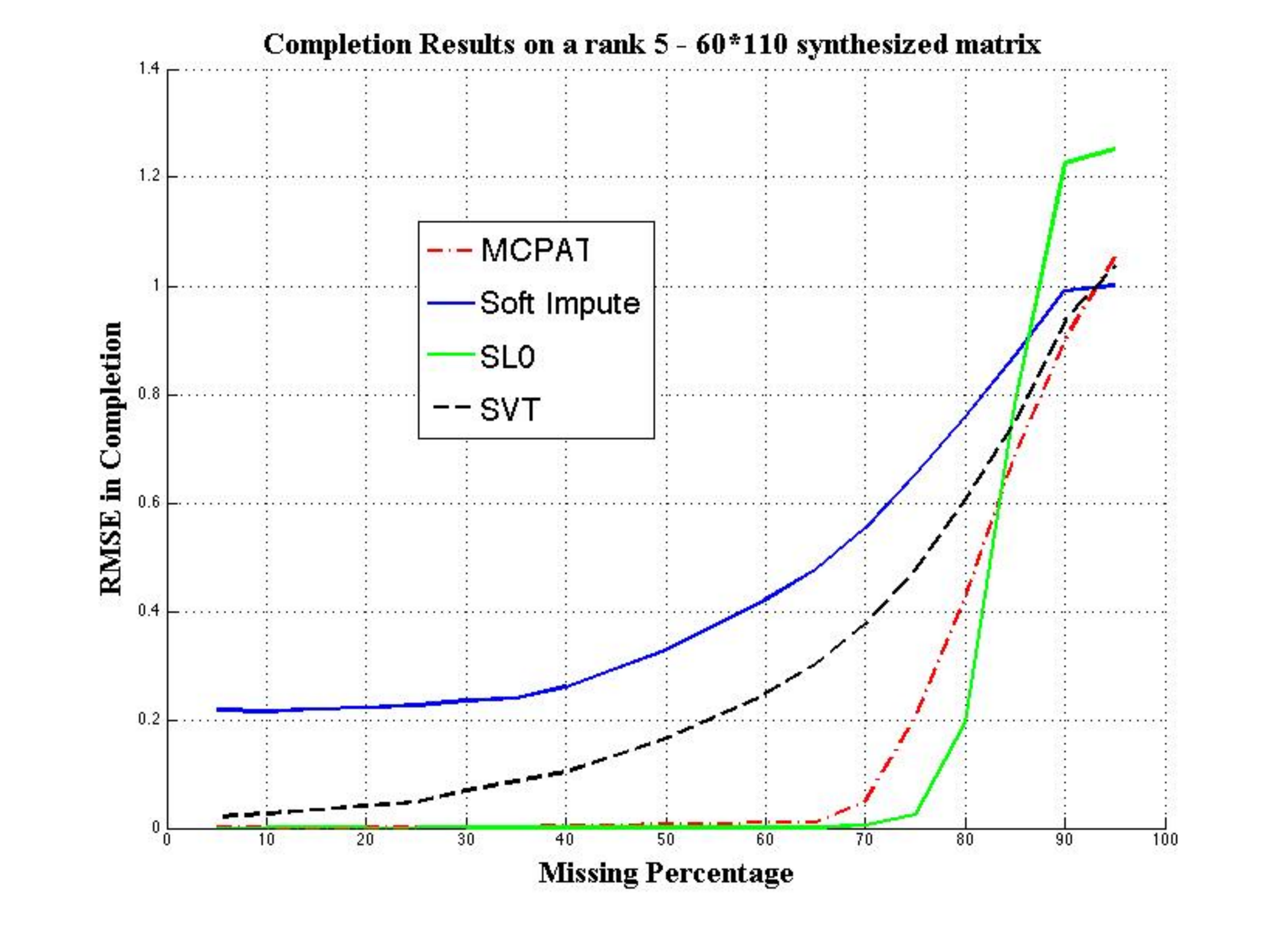}
	\caption{error comparison for matrix completion on a $60 \times 110$ matrix of rank $10$.}\label{fig:ashkan10}
\end{figure}
\begin{figure}
	\centering
	\includegraphics[width=1\linewidth]{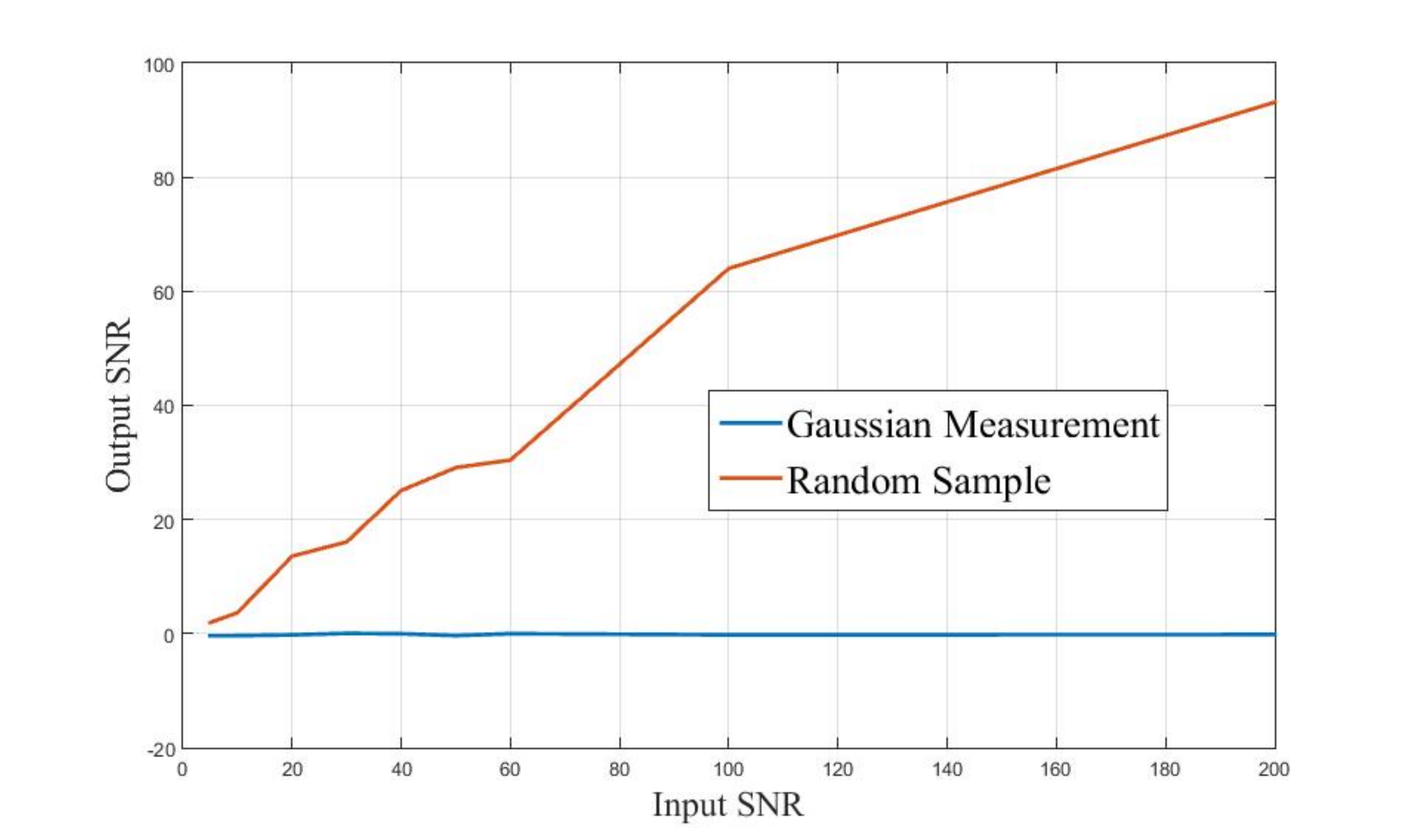}
	\caption{Comparison of Gaussian Measurement sensing matrix with Random Sampling method in reconstruction for 25 percent measurements}\label{fig:ashkan11}
\end{figure}
\begin{figure}
	\centering
	\includegraphics[width=1\linewidth]{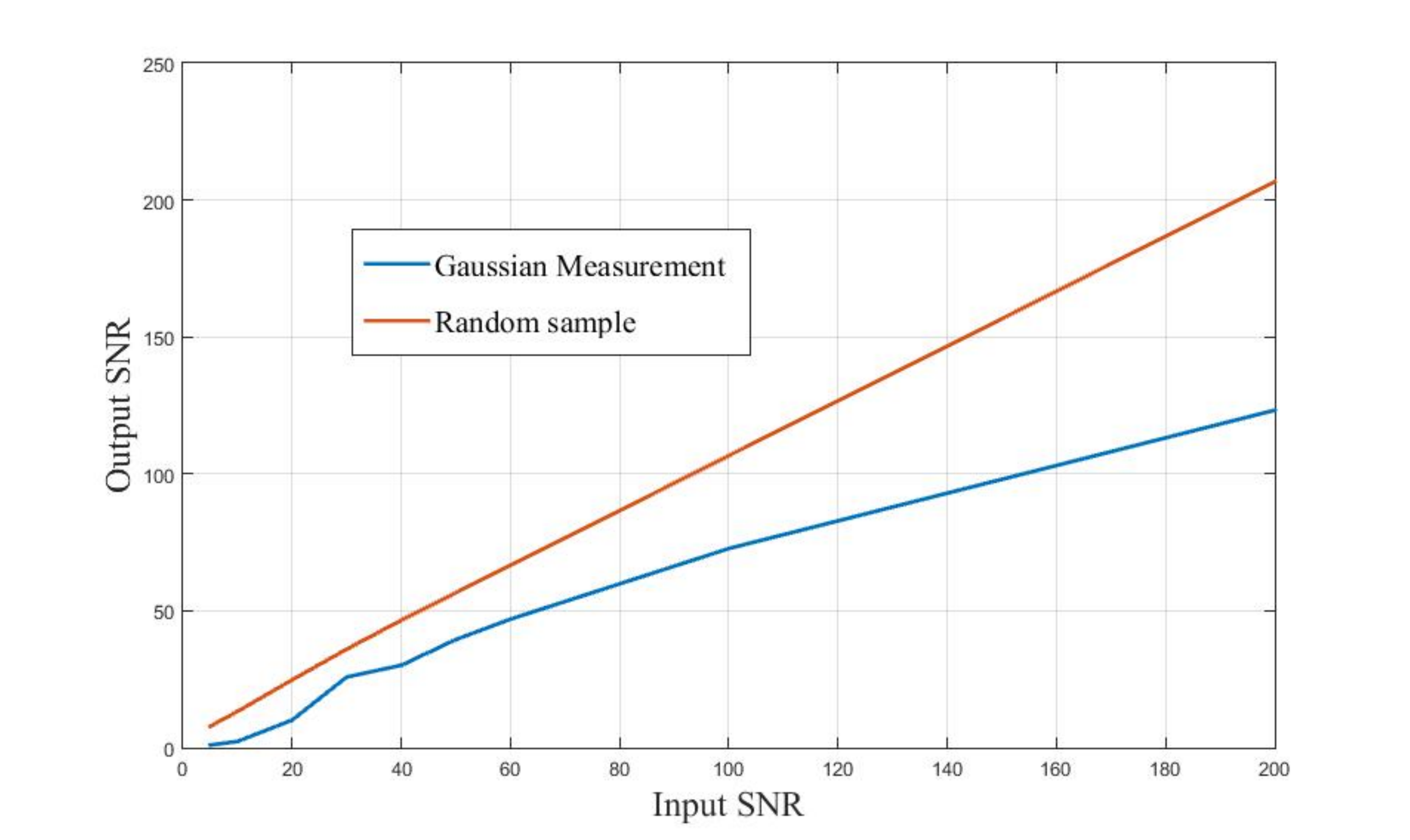}
	\caption{Comparison of Gaussian Measurement sensing matrix with Random Sampling method in reconstruction for 37.5 percent measurements}\label{fig:ashkan12}
\end{figure}
\begin{figure}
	\centering
	\includegraphics[width=1\linewidth]{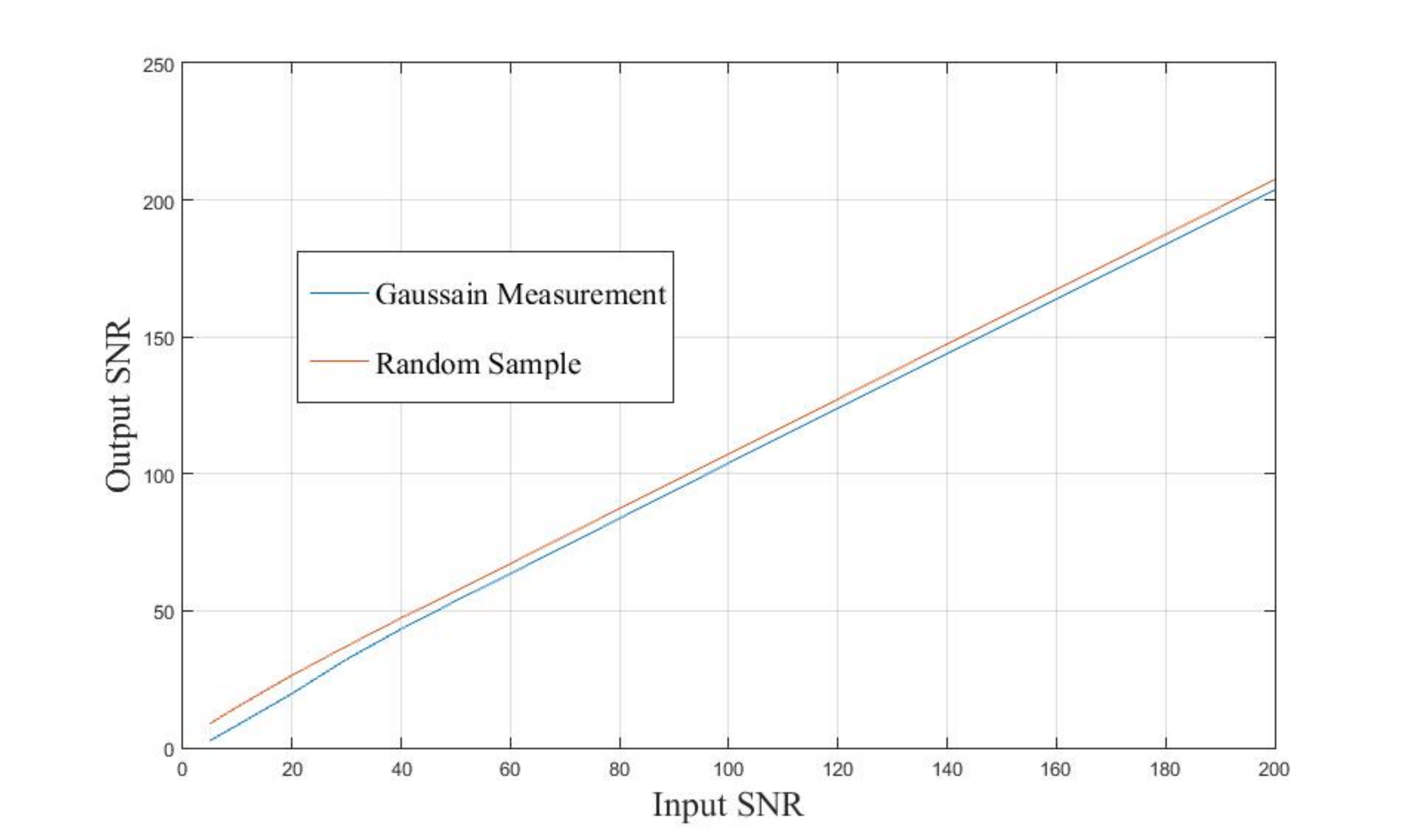}
	\caption{Comparison of Gaussian Measurement sensing matrix with Random Sampling method in reconstruction for 50 percent measurements}\label{fig:ashkan13}
\end{figure}
\begin{figure}
	\centering
	\includegraphics[width=1\linewidth]{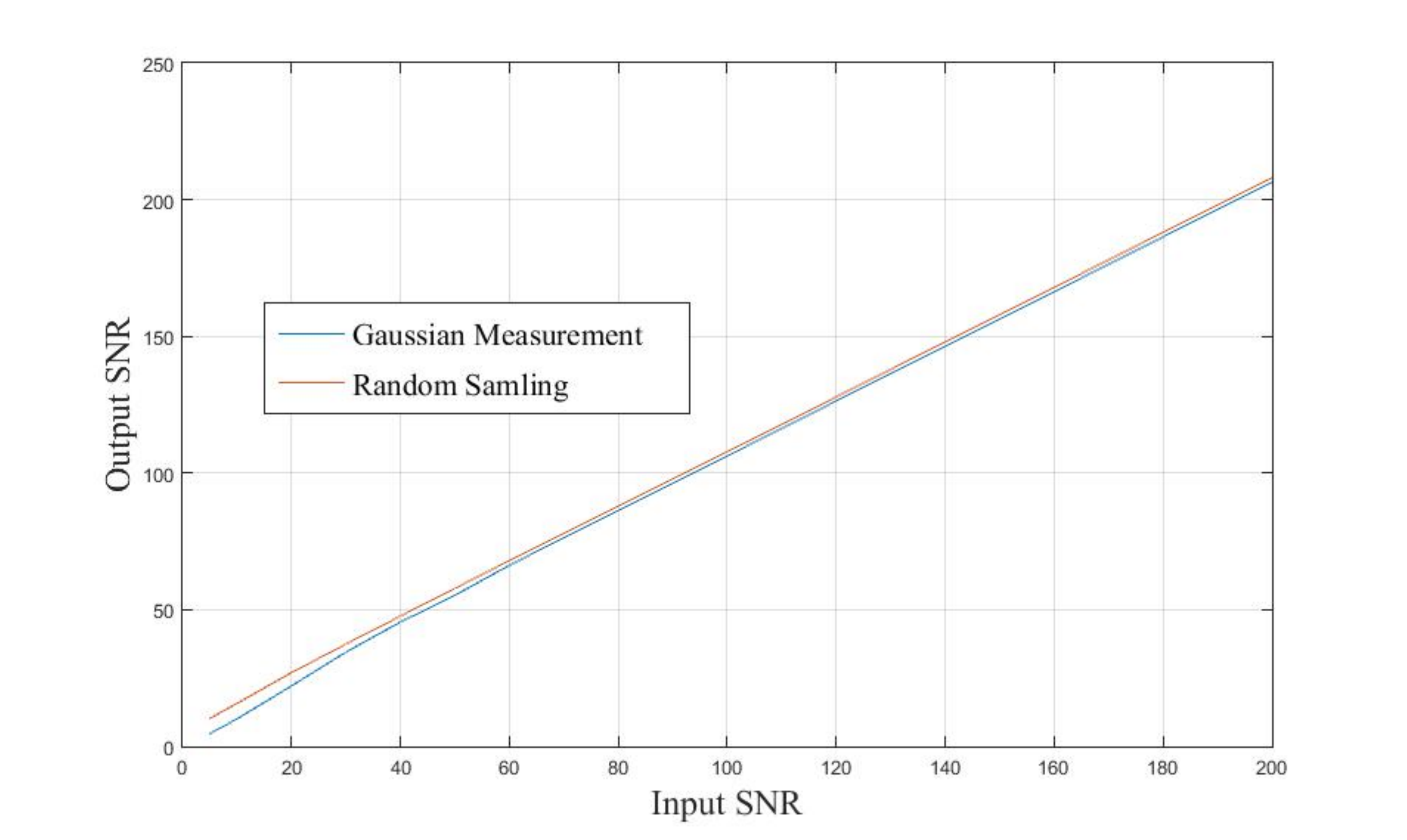}
	\caption{Comparison of Gaussian Measurement sensing matrix with Random Sampling method in reconstruction for 75 percent measurements}\label{fig:ashkan14}
\end{figure}

\section{Conclusion}  
We have introduced a new algorithm towards sparse signal recovery for CS, RS, and MC. The simulation results have shown that this sparse recovery method outperforms the well-known methods such as Lasso, OMP, SL0, and IMAT in terms of the SNR.
 A variation of this method which does not depend on computing pseudo-inverse at each step and works iteratively is more robust and faster in comparison to other mentioned methods. Simulation rsults imply that RS sensing measurement significantly outperforms the Gaussian meeasurement. The improvement is more pronounced at lower sampling rates. We have also observed that the extension of this idea to matrix completion also outperforms SVT, Soft-Impute, and SL0. 
\subsection*{Future Works}
We can modify the objective function we started working with as considering $l_1$-norm term for components outside the support instead of $l_2$-norm.
In order to minimize this objective function, we should solve $l_1$ Residual Shrinkage Minimization problem. The motivation to use this concept is that this formulation heuristically results in more resistance to noise. Thus, we do not have to regularize the initiation point as stated in section III, \rm{A}. Additionally, using $l_1$-norm guarantees more precision in signal recovery since $l_1$-norm is a better surrogate for $l_0$- norm than $l_2$ residual.  


\bibliographystyle{IEEEtran}
\bibliography{reffff.bib}
\end{document}